\documentclass[aps,pra,preprint,superscriptaddress,nofootinbib,longbibliography]{revtex4-2}

\usepackage{amsmath,amssymb,mathtools,amsthm}
\usepackage{bm}
\usepackage{graphicx}
\usepackage{microtype}
\usepackage{placeins}
\usepackage[colorlinks=true,linkcolor=blue,citecolor=blue,urlcolor=blue]{hyperref}
\hypersetup{
  pdftitle={Classical codes violate the conjectured square-root bound for quantum random access codes},
  pdfauthor={Kangqiao Liu}
}

\newtheorem{theorem}{Theorem}
\newtheorem{proposition}{Proposition}
\newtheorem{corollary}{Corollary}
\newenvironment{proofsketch}
  {\begin{proof}[Proof sketch]}
  {\end{proof}}

\newcommand{\bits}{\{0,1\}}
\newcommand{\Tr}{\operatorname{Tr}}
\newcommand{\hbin}{h_2}
\newcommand{\ket}[1]{|#1\rangle}
\newcommand{\bra}[1]{\langle #1|}
\newcommand{\proj}[1]{|#1\rangle \langle #1|}

\begin{document}

\title{Classical codes violate the conjectured square-root bound for quantum random access codes}

\author{Kangqiao Liu}
\email{kqliu@xhu.edu.cn}
\affiliation{School of Science, Xihua University, Chengdu 610039, China}
\affiliation{Key Laboratory of High Performance Scientific Computation, Xihua University, Chengdu 610039, China}

\begin{abstract}
We consider whether every quantum random access code (QRAC) with density-operator encodings and arbitrary decoding measurements obeys the conjectured bound $p\leq(1+\sqrt{m/n})/2$, where $n$ classical bits are encoded into $m$ qubits and $p$ is the worst-case success probability.  We find that classical random access codes with private randomness, which form a subclass of this QRAC model, violate the bound.  We embed these classical codes as QRACs with diagonal encoding states and commuting decoding measurements, and construct pure-state realizations with identical decoding statistics.  The achievability theorem of Ambainis, Nayak, Ta-Shma, and Vazirani then yields violations for every fixed $p\in(1/2,1)$ at sufficiently large input length.  The counterexamples span the full open interval between the conjectured and Nayak bounds at each fixed compression rate.  A finite-blocklength analysis further yields order-optimal logarithmic qubit scaling for a recovery bias scaling as $\sqrt{\log_2 n/n}$ with a sufficiently large prefactor.  These results identify the classical coding rate as the source of the separation and motivate restricted bounds based on quantitative spectral properties of decoding measurements.
\end{abstract}

\maketitle

\section{Introduction}
\label{sec:introduction}

Wiesner's quantum multiplexing protocol provides the earliest two-to-one instance of the selective-retrieval task underlying quantum random access codes (QRACs)~\cite{Wiesner1983}.  Ambainis, Nayak, Ta-Shma, and Vazirani (ANTV) introduced the general QRAC model, in which a sender compresses a string into one quantum message from which a receiver can later retrieve any chosen coordinate~\cite{AmbainisNayakTaShmaVazirani1999,AmbainisNayakTaShmaVazirani2002}.  Entropic QRAC bounds have supplied space lower bounds for one-way quantum finite automata and, through a quantum reduction, exponential lower bounds for two-query locally decodable codes~\cite{Nayak1999,KerenidisDeWolf2004}.  Operationally, QRAC decoding probabilities form dimension-constrained prepare-and-measure correlations.  This common structure allows observed statistics to certify a minimum Hilbert-space dimension~\cite{WehnerChristandlDoherty2008,GallegoBrunnerHadleyAcin2010} and supports semi-device-independent quantum key distribution and randomness certification~\cite{PawlowskiBrunner2011,LiEtAl2012}.  With parity-obliviousness imposed, the same task turns quantum advantage into an operational witness of preparation contextuality~\cite{SpekkensEtAl2009}.  Sharp QRAC success bounds therefore connect asymptotic information rates with finite-dimensional quantum correlations.

We study the binary worst-case formulation of this task.  Let $[n]=\{1,\ldots,n\}$.  Following Man\v{c}inska and Storgaard~\cite{MancinskaStorgaard2022}, an $(n,m,p)$ QRAC encodes each $x\in\bits^n$ as a density operator $\rho_x$ on a $2^m$-dimensional Hilbert space.  For every $i\in[n]$, a two-outcome positive operator-valued measure (POVM) $\{D_i^0,D_i^1\}$ is used to decode $x_i$, and the worst-case success condition is
\begin{equation}
 \min_{x\in\bits^n}\min_{i\in[n]}\Tr \left(\rho_xD_i^{x_i}\right)\ge p ,
 \label{eq:qrac-success}
\end{equation}
where $\Tr$ denotes the trace.
Throughout, we consider $1/2<p<1$ and $m<n$.  This general model permits mixed encoding states and decoding POVMs with mutually commuting effects.

For the shared-randomness model of Ref.~\cite{AmbainisLeungMancinskaOzols2008}, Man\v{c}inska and Storgaard proved $p\leq 1/2+\sqrt{2^{m-1}/n}/2$ and confirmed the two-qubit conjecture of Imamichi and Raymond~\cite{MancinskaStorgaard2022,ImamichiRaymond2018}.  They asked whether the stronger expression
\begin{equation}
 p\leq \frac12+\frac12\sqrt{\frac{m}{n}}
 \label{eq:conjbound}
\end{equation}
holds for arbitrary $(n,m,p)$ QRACs.  The conjectured bound is saturated by repeated concatenations of the optimal $(2,1)$ and $(3,1)$ QRACs, and it lies below Nayak's information-theoretic bound~\cite{Nayak1999,MancinskaStorgaard2022}.  Recent works use the conjectured bound as a reference for finite-size constructions~\cite{Suzuki2026,AkibueRaymondTamakiTeramoto2026}.  Kondo \textit{et al.} explicitly record its general validity as open~\cite{KondoEtAl2026}.

Prior work establishes the embedding of classical random access codes (RACs) into this general model.  Liab{\o}tr{\o} describes classical RACs as QRACs whose encoding states are diagonal in a fixed basis~\cite{Liabotro2017}.  Lin and de Wolf observe that the asymptotically optimal classical message length is achievable within this diagonal subclass~\cite{LinDeWolf2025}.  Kondo \textit{et al.} formulate stochastic classical RACs as diagonal QRACs~\cite{KondoEtAl2026}.

This paper constructs counterexamples to Eq.~\eqref{eq:conjbound} within the general QRAC model.  The construction combines this diagonal sector with the ANTV achievability theorem~\cite{AmbainisNayakTaShmaVazirani1999,AmbainisNayakTaShmaVazirani2002}.  Let $\hbin$ denote the binary entropy and write $s=2p-1$ for the recovery bias.  At fixed $p$, ANTV give classical RACs with private randomness, deterministic decoders, and message length $(1-\hbin(p))n+O(\log_2 n)$ bits.  The binary entropy obeys
\begin{equation}
 1-\hbin(p)<s^2 .
 \label{eq:entropy-gap-intro}
\end{equation}
The strict gap implies that the embedded ANTV codes violate Eq.~\eqref{eq:conjbound} for all sufficiently large $n$.  At each fixed compression rate, these codes fill the open interval above Eq.~\eqref{eq:conjbound} shown in Fig.~\ref{fig:window}; Nayak's bound $m\geq(1-\hbin(p))n$~\cite{Nayak1999} supplies its asymptotic upper endpoint.

The analysis establishes a finite-alphabet embedding and an exact classical characterization of QRACs with commuting decoding POVMs.  The same decoding statistics admit realizations with diagonal encoding states and with pure encoding states.  Strict violations persist in relative open neighborhoods and occur for projective decoding measurements whose effects are noncommuting across every pair of distinct decoding indices.  The finite-blocklength analysis also yields counterexamples with logarithmic qubit count in a vanishing-bias regime.  Thus the conjectured bound is incompatible with classical RACs with private randomness already included in the general QRAC model.

The remainder of the paper is organized as follows.  Section~\ref{sec:embedding} formulates the classical RAC embedding and characterizes QRAC statistics generated by commuting decoding POVMs.  The ANTV rate is applied in Sec.~\ref{sec:antv-counterexamples} to obtain the fixed-bias, finite-blocklength, and logarithmic-qubit counterexample families.  A compression-rate parametrization in Sec.~\ref{sec:rate-window} determines the full asymptotic counterexample window.  Robustness of the construction under perturbations of states and decoding measurements is established in Sec.~\ref{sec:robustness}.  Section~\ref{sec:examples} illustrates the formal results in the $(r,p)$ plane and through explicit finite-blocklength parameters.  Their asymptotic interpretation and relation to current QRAC bounds are developed in Sec.~\ref{sec:discussion}.  Complete proofs of all formal results and finite-blocklength details are collected in Appendix~\ref{app:proofs}.

\section{Embedding classical RACs}
\label{sec:embedding}

The construction begins with a precise formulation of the classical sector contained in the QRAC model.  A stochastic classical RAC with private randomness uses a finite message alphabet $\mathcal C$ and assigns a message distribution $P(c|x)$ to each input $x\in\bits^n$.  Given a message $c$ and a requested index $i\in[n]$, the decoder outputs $b\in\bits$ with probability $q_i(b|c)$, where $q_i(b|c)\geq0$ and $\sum_{b\in\bits}q_i(b|c)=1$.  Its worst-case success condition is
\begin{equation}
 \sum_{c\in\mathcal C}P(c|x)q_i(x_i|c)\ge p,
 \quad x\in\bits^n,
 \quad i\in[n].
 \label{eq:classical-success-general}
\end{equation}
The conditional distributions $P(c|x)$ specify the stochastic encoder.  A deterministic decoder is given by decoding functions $d_i:\mathcal C\to\bits$ through $q_i(b|c)=\mathbf 1\{d_i(c)=b\}$, where $\mathbf 1$ denotes the indicator function.

This setup leads to the embedding and commuting-POVM characterization below.

\begin{proposition}[Classical RAC embedding]
\label{prop:embedding}
Any classical RAC with private randomness on $n$ input bits, finite message alphabet $\mathcal C$, and worst-case success at least $p$ yields two $M=\lceil\log_2 |\mathcal C|\rceil$-qubit QRAC realizations with the same decoding statistics.  One has diagonal encoding states and the other has pure encoding states.  Deterministic decoders yield projective decoding measurements.  Both realizations violate Eq.~\eqref{eq:conjbound} whenever
\begin{equation}
 \frac{M}{n}<s^2 .
 \label{eq:general-criterion}
\end{equation}
\end{proposition}

\begin{proofsketch}
Embedding $P(c|x)$ and $q_i(b|c)$ as diagonal entries of encoding states and POVM effects makes the Born rule reproduce the classical decoding probabilities, with deterministic decoders giving projective decoding measurements through $\{0,1\}$-valued effect eigenvalues.  Replacing $P(c|x)$ by amplitudes $\sqrt{P(c|x)}$ gives the pure-state realization, and Eq.~\eqref{eq:general-criterion} is equivalent to a strict violation of Eq.~\eqref{eq:conjbound}.
\end{proofsketch}

\begin{proposition}[Classical characterization of commuting decoding POVMs]
\label{prop:converse}
Every $(n,M,p)$ QRAC for which all effects of the decoding POVMs commute has the same decoding statistics as a classical RAC with private randomness and a message alphabet of size at most $2^M$.
\end{proposition}

\begin{proofsketch}
Simultaneous diagonalization turns the diagonal entries of each encoding state into $P(c|x)$ and the effect eigenvalues into $q_i(b|c)$.  The resulting classical model reproduces every decoding probability.
\end{proofsketch}

Together, Propositions~\ref{prop:embedding} and \ref{prop:converse} characterize the QRAC statistics generated by commuting decoding POVMs as those of classical RACs with private randomness and at most $2^M$ messages.  Equation~\eqref{eq:general-criterion} therefore reduces the counterexample construction in this sector to a classical message-rate comparison.

\section{Counterexamples from the ANTV rate}
\label{sec:antv-counterexamples}

The ANTV achievability theorem supplies the classical message rate required by Eq.~\eqref{eq:general-criterion}.  Theorem 2.2 of Ref.~\cite{AmbainisNayakTaShmaVazirani1999} (see also Theorem 5.2 of Ref.~\cite{AmbainisNayakTaShmaVazirani2002}) gives classical RACs with private randomness on $n$ input bits, deterministic decoders, worst-case success at least $p$, and message length $(1-\hbin(p))n+O(\log_2 n)$ bits.  Let $k_n$ denote this bit length.  The ANTV construction obeys the finite-blocklength upper bound
\begin{equation}
 k_n\leq \left\lceil (1-\hbin(p))n+7\log_2 n\right\rceil
 \label{eq:ANTV-length}
\end{equation}
for all sufficiently large $n$.  Proposition~\ref{prop:embedding} embeds the $k_n$-bit messages into $k_n$ qubits and supplies diagonal-state and pure-state QRAC realizations with the same commuting projective decoding measurements.  Combining the ANTV rate with this embedding yields the fixed-bias counterexample family stated next.

\begin{theorem}[Fixed-bias counterexamples]
\label{thm:fixed-p}
For every fixed $p$ there is an $n_0(p)$ such that, for every $n\ge n_0(p)$, there exists an $(n,k_n,p)$ QRAC with $p>1/2+\sqrt{k_n/n}/2$, where $k_n<n$ satisfies Eq.~\eqref{eq:ANTV-length}.
\end{theorem}

\begin{proofsketch}
For fixed $p$, Eq.~\eqref{eq:entropy-gap-intro} and $\log_2 n/n\to0$ eventually give $k_n/n<s^2$.  Proposition~\ref{prop:embedding} then yields the counterexample.
\end{proofsketch}

For explicit parameters, this construction yields the following criterion.

\begin{proposition}[Finite-blocklength counterexamples]
\label{prop:finite-block}
Let $n$ satisfy $p\leq1-1/n$, and define
\begin{equation}
 K_{p,n}:=\left\lceil(1-\hbin(p))n+7\log_2 n\right\rceil .
 \label{eq:finite-antv-bound}
\end{equation}
If $K_{p,n}<ns^2$, there exists an $(n,K_{p,n},p)$ QRAC that violates Eq.~\eqref{eq:conjbound}.
\end{proposition}

\begin{proofsketch}
The finite-blocklength ANTV construction produces a classical RAC with deterministic decoders whose message alphabet embeds into $K_{p,n}$ bits.  Proposition~\ref{prop:embedding} converts it into the claimed QRAC, and the strict rate condition ensures the violation.
\end{proofsketch}

Applying Proposition~\ref{prop:finite-block} with a bias that vanishes with $n$ yields the logarithmic-qubit refinement.

\begin{corollary}[Logarithmic qubit count]
\label{cor:logarithmic}
Let $a=1-1/(2\ln 2)$, fix $C>7/a=25.120896\ldots$, and set the recovery-bias sequence $s_n=\sqrt{C\log_2 n/n}$ and $p_n=(1+s_n)/2$.
For all sufficiently large $n$, there exists an $(n,K_{p_n,n},p_n)$ QRAC that violates Eq.~\eqref{eq:conjbound} with $K_{p_n,n}=\Theta(\log_2 n)$, where $\Theta$ denotes asymptotic order.
\end{corollary}

\begin{proofsketch}
The entropy expansion together with $C>7/a$ eventually yields $K_{p_n,n}<ns_n^2$, so Proposition~\ref{prop:finite-block} applies.  The expansion gives $(C/(2\ln 2)+7)\log_2 n+o(\log_2 n)$, while Nayak's bound supplies a lower bound of the same asymptotic order.
\end{proofsketch}

The preceding results describe the construction through recovery bias and input length.  Reparameterizing the asymptotic family by its compression rate organizes these counterexamples into a region in the $(r,p)$ plane.

\section{Rate-window formulation}
\label{sec:rate-window}

For a fixed asymptotic compression rate $r\in(0,1)$, consider QRAC families satisfying $m/n\to r$.  The conjectured bound defines
\begin{equation}
 p_{\rm sq}(r)=\frac12+\frac12\sqrt r .
 \label{eq:psq}
\end{equation}
The asymptotic Nayak bound is
\begin{equation}
 p_{\rm N}(r)=\hbin^{-1}(1-r),
 \label{eq:pN}
\end{equation}
where the inverse is taken on $[1/2,1]$.  Equation~\eqref{eq:entropy-gap-intro} gives $p_{\rm sq}(r)<p_{\rm N}(r)$ throughout this rate range.  The open interval between these curves is
\begin{equation}
 p_{\rm sq}(r)<p<p_{\rm N}(r).
 \label{eq:open-window}
\end{equation}
Parameterizing the same ANTV family by $r$ and padding to $\lceil rn\rceil$ yields the following rate-window corollary.

\begin{corollary}[Rate window]
\label{cor:window}
Fix $r\in(0,1)$ and $p$ satisfying Eq.~\eqref{eq:open-window}.  Then, for all sufficiently large $n$, there is an $(n,M_n,p)$ QRAC with
\begin{equation}
 M_n=\lceil rn\rceil
 \label{eq:exact-rate-Mn}
\end{equation}
and with $p>1/2+\sqrt{M_n/n}/2$.
\end{corollary}

\begin{proofsketch}
The upper inequality in Eq.~\eqref{eq:open-window} places the ANTV message length below $M_n$ asymptotically and permits padding to $M_n$ bits.  The lower inequality gives $M_n/n<s^2$ after rounding, so Proposition~\ref{prop:embedding} yields the stated QRACs.
\end{proofsketch}

The rate-window corollary identifies the full asymptotic counterexample region.  The realization structure determines how these strict violations persist under perturbations.

\section{Structural origin and robustness}
\label{sec:robustness}

For an $M$-qubit embedded realization, the diagonal density operators form a simplex whose $2^M$ vertices correspond to the computational-basis messages.  Each classical message distribution maps to a point in this simplex.  Taking square-root amplitudes gives a pure-state representation with the same computational-basis statistics.  For deterministic decoders, each decoding function defines a diagonal projective measurement, so the decoding statistics arise from a one-way classical protocol with messages of at most $M$ bits.  The separation from the conjectured bound follows from the achievable rate of classical RACs with private randomness.

\begin{proposition}[Robustness]
\label{prop:robustness}
Let $\rho_x$, $x\in\bits^n$, be the encoding states of an $(n,M,p)$ QRAC, and let $D_i^b$, $i\in[n]$ and $b\in\bits$, be its decoding POVM effects.  Define its gap by $\delta:=p-1/2-\sqrt{M/n}/2$ and assume $\delta>0$.  Every family of valid encoding states $\widetilde\rho_x$ and decoding POVM effects $\widetilde D_i^b$ on the same $2^M$-dimensional Hilbert space satisfying
\begin{equation}
 \max_{x,i,b}\left[
 \frac12\left\|\widetilde\rho_x-\rho_x\right\|_1
 +\left\|\widetilde D_i^b-D_i^b\right\|_\infty
 \right]<\delta
 \label{eq:robustness-neighborhood}
\end{equation}
also violates Eq.~\eqref{eq:conjbound}, where $\|\cdot\|_1$ and $\|\cdot\|_\infty$ denote the trace norm and operator norm.  If all $D_i^b$ are nontrivial projectors and $n\geq2$, every operator-norm neighborhood of these decoding measurements contains projective decoding measurements satisfying $\|[\widetilde D_i^b,\widetilde D_j^{b'}]\|_\infty>0$ for all $i\ne j$ and $b,b'\in\bits$.  The projective decoding measurements obtained by embedding the ANTV codes satisfy this nontriviality condition.
\end{proposition}

\begin{proofsketch}
Trace-distance and operator-norm continuity preserve the strict success gap under Eq.~\eqref{eq:robustness-neighborhood}.  Worst-case success makes every ANTV decoding function attain both output values, and independent small unitary conjugations of nontrivial projectors avoid the finite union of real-analytic commutator zero sets.
\end{proofsketch}

Each embedded pure-state ANTV realization therefore lies in a relative open subset of valid state-POVM tuples that violate Eq.~\eqref{eq:conjbound}, and every operator-norm neighborhood of its decoding measurements contains a projective tuple whose effects are noncommuting across every pair of distinct decoding indices.

The formal results admit direct illustrations in the $(r,p)$ plane and at finite blocklength.

\section{Examples and parameter-space illustration}
\label{sec:examples}

Figure~\ref{fig:window} displays the full open rate window of Corollary~\ref{cor:window}.  Its markers correspond to the finite-blocklength guarantees from Proposition~\ref{prop:finite-block} listed in Table~\ref{tab:finite}.

\begin{figure}[t]
\includegraphics[width=\columnwidth]{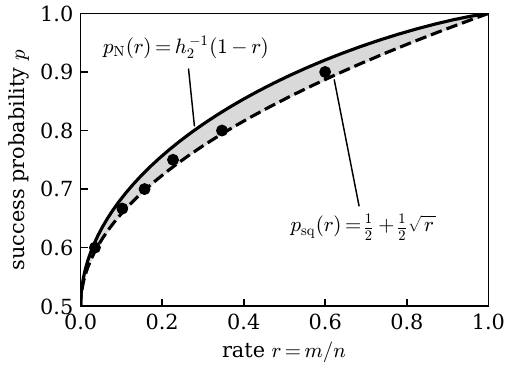}
\caption{Counterexamples fill the asymptotic interval between the conjectured and Nayak bounds.  For every rate $r$ and every $p$ in the shaded open window, $(n,\lceil rn\rceil,p)$ QRACs exist for all sufficiently large $n$ by Corollary~\ref{cor:window}.  The dashed and solid curves show $p_{\rm sq}(r)$ and $p_{\rm N}(r)$, respectively.  Markers show the parameter values of Table~\ref{tab:finite}.}
\label{fig:window}
\end{figure}

A representative finite-blocklength example illustrates the guarantee of Proposition~\ref{prop:finite-block}.
For $p=3/4$ and $n=4096$, Eq.~\eqref{eq:finite-antv-bound} gives
\begin{align}
 m=K_{3/4,4096}&=\left\lceil (1-\hbin(3/4))4096+7\log_2 4096\right\rceil \nonumber \\
  &=858 .
\end{align}
At this success probability, $s=1/2$.  The hypotheses of Proposition~\ref{prop:finite-block} hold because $3/4\leq1-1/4096$ and $858<4096s^2=1024$.  The proposition therefore gives a $(4096,858,3/4)$ QRAC with
\begin{equation}
 p_{\rm sq}(858/4096)=0.7288409143\ldots<\frac34 .
\end{equation}

Table~\ref{tab:finite} extends this calculation to six additional parameter choices.

\begin{table}[t]
\caption{Finite-blocklength QRAC existence guarantees from Proposition~\ref{prop:finite-block}.  Each row sets $m=K_{p,n}$ and lists the gap above $p_{\rm sq}(m/n)$.}
\label{tab:finite}
\begin{ruledtabular}
\begin{tabular}{cccc}
$p$ & $n$ & $m$ & $p-p_{\rm sq}(m/n)$ \\
\hline
$0.6000$ & $16384$ & $574$ & $0.006413$ \\
$2/3$ & $4096$ & $419$ & $0.006749$ \\
$0.7000$ & $2048$ & $321$ & $0.002049$ \\
$0.7500$ & $2048$ & $464$ & $0.012007$ \\
$0.8000$ & $1024$ & $355$ & $0.005602$ \\
$0.9000$ & $1024$ & $614$ & $0.012828$

\end{tabular}
\end{ruledtabular}
\end{table}
\FloatBarrier

\section{Discussion and conclusion}
\label{sec:discussion}

The combined achievability and converse results determine the asymptotic meaning of the two curves in Fig.~\ref{fig:window}.  At every fixed rate $r$, the ANTV construction approaches $p_{\rm N}(r)$ from below, and Nayak's bound supplies the matching converse.  The entropy curve $p_{\rm N}(r)$ therefore gives the asymptotic supremum of the worst-case success probability in the unrestricted QRAC model.  The square-root curve $p_{\rm sq}(r)$ is attained by several structured families and lies strictly inside the feasible region at every fixed rate.  Thus the general density-operator model is governed asymptotically by the entropy rate of its embedded classical message sector.

The realizations identify structural features that preserve this separation.  For the paired realizations of Proposition~\ref{prop:embedding}, the diagonal mixed states and pure states have identical decoding statistics and the same embedded classical message rate.  Proposition~\ref{prop:robustness} places projective decoding measurements with noncommuting effects across every pair of distinct decoding indices in every operator-norm neighborhood of the commuting realization.  The success gap is stable under perturbations that change the measurement commutators.  Measurement incompatibility is known to be necessary for quantum advantage over classical RAC strategies under an average-success comparison~\cite{CarmeliHeinosaariToigo2020}.  Relevant quantitative parameters include decoder spectra, measurement overlaps, and the minimum classical message-alphabet size required to reproduce the decoding statistics.  This interpretation complements the spectral formulation of Akibue \textit{et al.}, whose finite-size bounds and equality conditions are expressed through decoding-measurement geometry~\cite{AkibueRaymondTamakiTeramoto2026}.

The high-rate constructions of Suzuki and Kondo \textit{et al.} show that the square-root value is attained within specified families.  Suzuki establishes average-case saturation for $(n,n-1)$ QRACs~\cite{Suzuki2026}.  Tan and Jafar prove the conjectured average-success upper bound for the $(n,n-1)$ and $(n,n-2)$ cases~\cite{TanJafar2026}.  Kondo \textit{et al.} recover related high-rate QRACs and develop geometric characterizations of classical RACs under average- and worst-case criteria~\cite{KondoEtAl2026}.  The rate-window corollary places these high-rate results in a broader parameter space, where classical RACs with private randomness reach every fixed-rate point below the entropy boundary.  The finite-parameter classical-quantum gaps reported by Kondo \textit{et al.} and the present asymptotic classical achievability together show that RAC-QRAC comparisons depend on the scaling regime and success criterion.

Success criterion also connects the present results to general QRAC bounds.  Farkas, Miklin, and Tavakoli bound average success for arbitrary input alphabet and Hilbert-space dimension~\cite{FarkasMiklinTavakoli2025}.  Every worst-case guarantee obtained here is also an average-case guarantee, so the embedded ANTV codes provide explicit feasible points for that broader optimization problem.  Their framework constrains average-case performance at finite dimension.  The present construction supplies asymptotically tight achievable worst-case rates in the binary setting.

The counterexample families have $m\to\infty$, with the logarithmic construction showing that growth of the embedded classical message space already produces a separation at $m=\Theta(\log_2 n)$.  Existing fixed-$m$ theorems and finite small-block constructions retain their stated parameter domains.  Natural restricted formulations may combine fixed-$m$ regimes with quantitative spectral regularity of the decoding measurements.  Identifying quantitative conditions that contain the known square-root-saturating families and control the embedded classical rate remains a structural question.  Within the standard density-operator model, the asymptotic worst-case boundary is $p_{\rm N}(r)$, and robust counterexamples occur throughout the open region above $p_{\rm sq}(r)$.

\appendix
\section{Detailed proofs}
\label{app:proofs}

\begin{proof}[Proof of Proposition~\ref{prop:embedding}]
The choice $M=\lceil\log_2 |\mathcal C|\rceil$ permits an injection of $\mathcal C$ into the computational basis.  Let $I$ denote the identity operator and define
\begin{equation}
 \Pi_\perp=I-\sum_{c\in\mathcal C}\proj c,
\end{equation}
where $\Pi_\perp$ projects onto the unused basis states.  Define
\begin{align}
 \rho_x&=\sum_{c\in\mathcal C}P(c|x)\proj c,
 \label{eq:diagonal-state}\\
 D_i^b&=\sum_{c\in\mathcal C}q_i(b|c)\proj c
       +\mathbf 1\{b=0\}\Pi_\perp,
 \quad b\in\bits.
 \label{eq:diagonal-povm-general}
\end{align}
Each $\rho_x$ is positive semidefinite and has unit trace.  The normalization of $q_i(b|c)$ gives
\begin{equation}
 D_i^0+D_i^1=\sum_{c\in\mathcal C}\proj c+\Pi_\perp=I .
\end{equation}
The coefficients of the basis projectors lie in $[0,1]$, so $0\leq D_i^b\leq I$.  All POVM effects are diagonal and hence commute.  Direct evaluation gives
\begin{equation}
 \Tr(\rho_xD_i^{x_i})=
 \sum_{c\in\mathcal C}P(c|x)q_i(x_i|c).
 \label{eq:embedding-equality-general}
\end{equation}

For a deterministic decoder, every coefficient in Eq.~\eqref{eq:diagonal-povm-general} lies in $\{0,1\}$.  The two effects are then complementary orthogonal projectors.  Define
\begin{equation}
 \ket{\psi_x}=\sum_{c\in\mathcal C}\sqrt{P(c|x)}\ket c
 \label{eq:coherent-state}
\end{equation}
The vectors are normalized, and diagonality gives
\begin{equation}
 \bra{\psi_x}D_i^{x_i}\ket{\psi_x}
 =\sum_{c\in\mathcal C}P(c|x)q_i(x_i|c).
 \label{eq:pure-equality}
\end{equation}
Thus the realizations with diagonal encoding states and pure encoding states have the same decoding statistics.  Finally, Eq.~\eqref{eq:general-criterion} is equivalent to $1/2+\sqrt{M/n}/2<p$.
\end{proof}

\begin{proof}[Proof of Proposition~\ref{prop:converse}]
The commuting Hermitian effects admit a common orthonormal eigenbasis $\{\ket c\}_{c=1}^{2^M}$.  In this basis, define
\begin{equation}
 P(c|x)=\bra c\rho_x\ket c,
 \qquad
 q_i(b|c)=\bra cD_i^b\ket c .
\end{equation}
Positivity and normalization of the states make $P(\cdot|x)$ a probability distribution.  The POVM relations make $q_i(\cdot|c)$ a probability distribution.  Since each $D_i^b$ is diagonal in the common basis,
\begin{equation}
 \Tr(\rho_xD_i^b)=
 \sum_{c=1}^{2^M}P(c|x)q_i(b|c),
\end{equation}
which reproduces all decoding probabilities with a classical message alphabet of size $2^M$.
\end{proof}

\begin{proof}[Derivation of Eq.~\eqref{eq:entropy-gap-intro}]
Since $s\in(0,1)$, the binary entropy deficit has the convergent expansion
\begin{equation}
 1-\hbin \left(\frac{1+s}{2}\right)
 =\frac{1}{\ln 2}\sum_{j=1}^{\infty}
 \frac{s^{2j}}{(2j-1)(2j)} .
 \label{eq:entropy-series}
\end{equation}
The coefficient sum satisfies
\begin{equation}
 \frac{1}{\ln 2}\sum_{j=1}^{\infty}
 \frac{1}{(2j-1)(2j)}=1,
 \label{eq:series-one}
\end{equation}
because $[(2j-1)(2j)]^{-1}=(2j-1)^{-1}-(2j)^{-1}$ and the resulting series sums to $\ln 2$.  In this range, every $s^{2j}$ is at most $s^2$, and the inequality is strict for $j\geq2$.  Equations~\eqref{eq:entropy-series} and \eqref{eq:series-one} therefore give
\begin{equation}
 1-\hbin \left(\frac{1+s}{2}\right)<s^2 .
 \label{eq:gap-proof}
\end{equation}
Since $p=(1+s)/2$, this yields Eq.~\eqref{eq:entropy-gap-intro}.
\end{proof}

\begin{proof}[Proof of Theorem~\ref{thm:fixed-p}]
For fixed $p$, define
\begin{equation}
 \Delta_p:=s^2-[1-\hbin(p)]>0,
\end{equation}
where positivity follows from Eq.~\eqref{eq:entropy-gap-intro}.  Since $(7\log_2 n+1)/n\to0$, it is smaller than $\Delta_p$ for all sufficiently large $n$.  In this range, the ANTV classical RAC exists with the length in Eq.~\eqref{eq:ANTV-length} and satisfies
\begin{equation}
 \frac{k_n}{n}
 \leq 1-\hbin(p)+\frac{7\log_2 n+1}{n}
 <s^2<1 .
\end{equation}
Proposition~\ref{prop:embedding} yields the claimed QRAC and the strict violation of Eq.~\eqref{eq:conjbound}.
\end{proof}

\begin{proof}[Proof of Proposition~\ref{prop:finite-block}]
For $p\leq1-1/n$, the finite-blocklength ANTV construction uses a covering-code index of at most $(1-\hbin(p))n+4\log_2 n$ bits and an index for $n^3$ permutation-mask pairs of length $3\log_2 n$.  The total message length is therefore at most $(1-\hbin(p))n+7\log_2 n$ bits.  The decoders are deterministic.  A Chernoff bound of $e^{-2n}$ applies to each pair $(x,i)$, and the union bound over the $n2^n$ pairs gives total failure probability at most $n2^ne^{-2n}<1$ for every positive integer $n$.  Hence a fixed family of permutation-mask pairs attains the required worst-case success.

Consequently, the message alphabet obeys
\begin{equation}
 \left\lceil\log_2 |\mathcal C|\right\rceil
 \leq K_{p,n} .
\end{equation}
Padding to $K_{p,n}$ bits preserves the decoding statistics.  Under $K_{p,n}<ns^2$, Proposition~\ref{prop:embedding} gives the claimed QRAC and the strict violation of Eq.~\eqref{eq:conjbound}.
\end{proof}

\begin{proof}[Proof of Corollary~\ref{cor:logarithmic}]
Equation~\eqref{eq:entropy-series} gives
\begin{equation}
 \Delta_{p_n}:=s_n^2-[1-\hbin(p_n)]
 =aC\frac{\log_2 n}{n}
 +O\!\left(\frac{(\log_2 n)^2}{n^2}\right).
\end{equation}
Since $aC>7$, $n\Delta_{p_n}-7\log_2 n\to\infty$, which gives $K_{p_n,n}<ns_n^2$ for all sufficiently large $n$.  In the same range, $p_n\leq1-1/n$, so Proposition~\ref{prop:finite-block} gives the stated QRAC.  The expansion
\begin{equation}
 1-\hbin(p_n)=\frac{s_n^2}{2\ln 2}+O(s_n^4)
\end{equation}
yields
\begin{equation}
 K_{p_n,n}=\left(\frac{C}{2\ln 2}+7\right)\log_2 n+o(\log_2 n).
\end{equation}
For any QRAC with success $p_n$, Nayak's bound~\cite{Nayak1999} gives
\begin{equation}
 m\geq[1-\hbin(p_n)]n
 =\left(\frac{C}{2\ln 2}+o(1)\right)\log_2 n .
\end{equation}
The construction and lower bound establish order-optimal logarithmic qubit scaling at this bias.
\end{proof}

\begin{proof}[Proof of Corollary~\ref{cor:window}]
The condition $p<\hbin^{-1}(1-r)$ is equivalent to $1-\hbin(p)<r$.  The ANTV length $k_n$ in Eq.~\eqref{eq:ANTV-length} is therefore at most $M_n$ for all sufficiently large $n$.  Padding each classical message with $M_n-k_n$ fixed zero bits and applying Proposition~\ref{prop:embedding} gives an $(n,M_n,p)$ QRAC.  The lower inequality in Eq.~\eqref{eq:open-window} gives $r<s^2$.  Hence $M_n/n<s^2$ for all sufficiently large $n$, and the padded QRAC violates Eq.~\eqref{eq:conjbound} at the exact qubit count $M_n$.
\end{proof}

\begin{proof}[Proof of Proposition~\ref{prop:robustness}]
For $0\leq\widetilde D_i^b\leq I$, the trace-distance and operator-norm bounds give
\begin{align}
 \Tr(\widetilde\rho_x\widetilde D_i^{x_i})
 \geq{}& \Tr(\rho_xD_i^{x_i})
 -\frac12\|\widetilde\rho_x-\rho_x\|_1 \notag\\
 &-\|\widetilde D_i^{x_i}-D_i^{x_i}\|_\infty .
 \label{eq:robustness-success}
\end{align}
Equation~\eqref{eq:robustness-neighborhood} therefore preserves a success probability greater than $p-\delta=1/2+\sqrt{M/n}/2$.

For each $i\ne j$, define the real-analytic map
\begin{equation}
 \mathcal F_{ij}(U_1,\ldots,U_n)
 =[U_iD_i^0U_i^\dagger,U_jD_j^0U_j^\dagger]
\end{equation}
on the product unitary group.  For any two nontrivial projectors $\Gamma$ and $\Lambda$, a unitary $V$ can be chosen such that $[\Gamma,V\Lambda V^\dagger]\ne0$.  Thus each $\mathcal F_{ij}$ takes a nonzero value.  The product unitary group is connected, and the zero set of each $\mathcal F_{ij}$ is a proper real-analytic subset with empty interior.  These zero sets are closed, so their finite union also has empty interior.  Every neighborhood of the identity tuple therefore contains a unitary tuple outside this union.  Continuity of unitary conjugation places the resulting tuple of decoding measurements in any prescribed operator-norm neighborhood of the original tuple.  Setting $\widetilde D_i^b=U_iD_i^bU_i^\dagger$ gives the stated commutator condition.  Since $D_i^1=I-D_i^0$, complementary binary effects change each commutator only by a sign.

For the embedded ANTV codes, fix $i\in[n]$ and $b\in\bits$, and choose an input $x$ with $x_i=b$.  The positive worst-case success probability implies $P(c|x)>0$ for some message $c$ satisfying $d_i(c)=b$.  Hence every decoding function $d_i$ attains both output values on $\mathcal C$.  Under the embedding of Proposition~\ref{prop:embedding}, both complementary projectors $D_i^0$ and $D_i^1$ consequently have ranks between $1$ and $2^M-1$.
\end{proof}

\bibliography{qrac_counterexample}

\end{document}